\pdfoutput =1
  
\documentclass[pra,onecolumn,superscriptaddress]{article}

\usepackage[T1]{fontenc}
\usepackage[utf8]{inputenc}
\usepackage{amsmath,amsfonts,amsthm,amssymb}
\usepackage{mathtools}
\usepackage{subeqnarray}
\usepackage{setspace}
\usepackage{graphics,graphicx,color}
\usepackage{url}
\usepackage{enumerate}
\usepackage{float} 
\usepackage{multirow}
\usepackage{hhline}
\usepackage[margin=1.2in]{geometry}
\usepackage{braket}
\usepackage{dsfont} 
\usepackage{authblk}
\usepackage{multirow}
\usepackage{hhline}
\usepackage[makeroom]{cancel}
\usepackage{appendix}

\newtheorem{theorem}{Theorem}
\newtheorem{cor}{Corollary}
\newtheorem{definition}{Definition}
\newtheorem{lem}{Lemma}

\newtheorem{conjecture}{Conjecture}

\usepackage[scr=boondoxo,scrscaled=1.05]{mathalfa}

\DeclarePairedDelimiter\floor{\lfloor}{\rfloor}

\newcommand{\beq}{\begin{eqnarray}}
\newcommand{\eeq}{\end{eqnarray}}

\newcommand{\B}{\mathcal{B}}
\newcommand{\tB}{t\mathcal{B}}

\begin{document}
\title{A generalization of the CHSH inequality self-testing maximally entangled states of any local dimension}

\author{Andrea Coladangelo}
\date{}

\affil{Computing and Mathematical Sciences,	
Caltech\\ 	
{acoladan@caltech.edu}
}

\maketitle

\abstract{For every $d \geq 2$, we present a generalization of the CHSH inequality with the property that maximal violation self-tests the maximally entangled state of local dimension $d$. This is the first example of a family of inequalities with this property. Moreover, we provide a conjecture for a family of inequalities generalizing the tilted CHSH inequalities, and we conjecture that for each pure bipartite entangled state there is an inequality in the family whose maximal violation self-tests it. All of these inequalities are inspired by the self-testing correlations of \cite{CGS17}.}


\section{Introduction}
The CHSH inequality \cite{CHSH69} is one of the most well-studied witnesses of Bell's theorem \cite{Bell1964}, separating classical from quantum correlations. It is well-known \cite{Werner88, Popescu92} that maximal violation of the CHSH inequality can be achieved uniquely by measuring a maximally entangled pair of qubits. This kind of characterization is often referred to in the literature as self-testing \cite{Mayers2004}.

In this note, for any $d \geq 2$, we present a generalization of the CHSH inequality with the property that maximal violation self-tests the maximally entangled state of local dimension $d$. This is the first example of a family of Bell inequalities (or non-local games) for which this property holds for any $d \geq 2$. Previously, through various results in parallel self-testing, we knew of families self-testing maximally entangled states of local dimension $d$ a power of $2$ \cite{Mckague2016, Coladangelo2016, Coudron2016, CRSV2016, Natarajan2016}, or $d^n$ for any $d$ and for any $n\geq 2$ even \cite{CS2017b}. Another family of Bell inequalities, indexed by $d\geq 2$, was proposed by Salavrakos et al. \cite{Salavrakos17}. There, the authors show that maximal violation is achieved by the maximally entangled state of local dimension $d$, but the self-testing property is only conjectured (although numerical evidence is given for the case $d=3$).
Ours is not the first generalization of the CHSH inequality (or the CHSH game): a more natural algebraic generalization of the CHSH game over fields of order $q$ was introduced by Buhrman and Massar \cite{Buhrman05}, and studied by Bavarian and Shor \cite{Bavarian15}; another generalization was introduced by Tavakoli et al. and studied in the context of random access codes \cite{Tavakoli15}. However, the self-testing properties of these generalizations are not known. Our generalization is unrelated to these, and in this note we focus on establishing its self-testing properties. 

Our Bell inequality is inspired directly by \cite{CGS17}. There, the authors show that for each pure bipartite entangled state there exists a quantum correlation that self-tests it. Here, our aim is to phrase this self-test in terms of maximal violation of some Bell inequality (instead of in terms of a correlation). In other words, we wish to find, for each such correlation, a Bell inequality whose maximal quantum violation is attained exclusively at that correlation, i.e. a hyperplane tangent to the set of quantum correlations only at that self-testing point. We succeed in the maximally entangled case, and this yields a generalization of the CHSH inequality. We also provide loose analytical bounds on the robustness of our self-test.

The approach generalizes naturally also to the tilted case, and we present a candidate family of Bell inequalities generalizing the family of tilted CHSH inequalities \cite{Acin12}. However, in the tilted case, the lack of symmetry seems to make the analysis surprisingly more complicated, and we can only conjecture that for each pure bipartite entangled state of any local dimension there is a corresponding inequality in the family whose maximal violation self-tests it.

The Bell operator for the maximally entangled case is presented in section \ref{sec: main}, and concisely stated in Definition \ref{def: bell-operators}. The quantum bound and self-testing results are stated in Theorems \ref{thm: quantum-bound} and \ref{thm: exact self-testing}. Section \ref{sec: conj} presents the candidate family for the tilted case, and the corresponding conjecture. In Appendix \ref{sec: appendix}, we give a robust version of our self-test for the maximally entangled case (Theorem \ref{thm: robust self test}).

\subsection{Notation and preliminaries}
\paragraph{Correlations and strategies}
Given sets $\mathcal{X}$,$\mathcal{Y}$, $\mathcal{A}$, $\mathcal{B}$, a (bipartite) \textit{correlation} is a collection of conditional probability distributions $\{p(a,b|x,y): a\in \mathcal{A}, b \in \mathcal{B}\}_{(x,y)\in \mathcal{X}\times \mathcal{Y}}$. $\mathcal{X}$ and $\mathcal{Y}$ are referred to as the question sets, while $\mathcal{A}$ and $\mathcal{B}$ as the answer sets.

Given question sets and answer sets $\mathcal{X}$, $\mathcal{Y}$, $\mathcal{A}$, $\mathcal{B}$, a \textit{classical strategy} is specified by an integer $k$, a probability distribution $\{\lambda_i\}$ on $\{1,..,k\}$, a probability distribution $\{p_{x,i}^{a}\}$ on $\mathcal{A}$ for each $x \in \mathcal{X}$ and $ 1\leq i \leq k$, and a probability distribution $\{q_{y,i}^{b}\}$ on $\mathcal{B}$ for each $y \in \mathcal{Y}$ and $ 1\leq i \leq k$. It produces the correlation $p$ such that
$$p(a,b|x,y) = \sum_{i=1}^{k} \lambda_i p_{x,i}^a q_{y,i}^b \,\,\,\,\forall a \in \mathcal{A}, b \in \mathcal{B}, x \in \mathcal{X}, y \in \mathcal{Y}.$$

Given question sets and answer sets $\mathcal{X}$, $\mathcal{Y}$, $\mathcal{A}$, $\mathcal{B}$,
a \textit{quantum strategy} is specified by Hilbert spaces $\mathcal{H}_A$ and  $\mathcal{H}_B$, a pure state $\ket{\psi} \in \mathcal{H}_A \otimes \mathcal{H}_B$, and projective measurements $\{\Pi^a_{A_x}\}_a$ on $\mathcal{H}_A$, $\{\Pi^b_{B_y}\}_b$ on $\mathcal{H}_B$, for $x \in \mathcal{X}, y \in \mathcal{Y}$. It produces the correlation $p$ such that $$p(a,b|x,y) = \bra{\psi} \Pi_{A_x}^a \otimes \Pi_{B_y}^b \ket{\psi} \,\,\,\forall a \in \mathcal{A}, b \in \mathcal{B}, x \in \mathcal{X}, y \in \mathcal{Y}.$$
Concisely, we refer to a quantum strategy as a triple $\left(\ket{\psi}, \{\Pi^a_{A_x}\}_a, \{\Pi^b_{B_y}\}_b \right)$. We take the measurements to be projective without loss of generality by appealing to Naimark's dilation theorem. We will sometimes describe a quantum strategy by specifying an observable for each question. The observables in turn specify the projectors through their eigenspaces. 

A correlation is said to be classical (quantum) if there exists a classical (quantum) strategy producing it. We denote by $\mathcal{C}_c^{m,n,r,s}$ and $\mathcal{C}_q^{m,n,r,s}$ respectively the sets of classical and quantum correlations on question sets of sizes $m, n$ and answer sets of sizes $r,s$. 

\paragraph{Self-testing} We define self-testing formally:
\begin{definition}[Self-testing]
We say that a correlation $\{p^*(a,b|x,y): a \in \mathcal{A}, b \in \mathcal{B}\}_{x \in \mathcal{X}, y \in \mathcal{Y}}$ self-tests a strategy $\left(\ket{\Psi}, \{\tilde{\Pi}_{A_x}^{a}\}_{a}, \{\tilde{\Pi}_{B_y}^{b}\}_{b} \right)$ if, for any strategy $\left(\ket{\psi}, \{\Pi_{A_x}^a\}_a, \{\Pi_{B_y}^b\}_b \right)$ achieving $p^*$, there exists a local isometry $\Phi= \Phi_A\otimes\Phi_B$ and an auxiliary state $\ket{aux}$ such that 
\begin{align}
\Phi(\ket{\psi}) &=  \ket{\Psi} \otimes \ket{\textit{aux}} \label{eq: state}\\ \label{eq: measurements}
\Phi(\Pi_{A_x}^{a}\otimes \Pi_{B_y}^{b}\ket{\psi}) &=  \tilde{\Pi}_{A_x}^{a}\otimes \tilde{\Pi}_{B_y}^{b} \ket{\Psi} \otimes \ket{\textit{aux}}  \,\,\, \forall a \in \mathcal{A} ,b \in \mathcal{B}, x \in \mathcal{X}, y \in \mathcal{Y}
\end{align}
\end{definition}
Sometimes, we refer to \textit{self-testing of the state} when we are only concerned with the guarantee of equation \eqref{eq: state}, and not \eqref{eq: measurements}. Moreover, we will informally say that ``maximal violation of an inequality self-tests a state'' to mean precisely that any correlation achieving maximal violation self-tests the state. There is also a notion of \textit{robust} self-testing, when one can approximately characterize strategies that are close to achieving the ideal correlation \cite{MYS12, Kaniewski16}. For a precise definition, we refer the reader to \cite{CS17}. We remark that, technically, we don't need to assume that the original strategy uses a pure state $\ket{\psi}$, but rather our proofs can be directly translated to the case of a mixed state (see \cite{CGS17} for a more precise account of this).

\paragraph{The family of tilted CHSH inequalities} We introduce the family of tilted CHSH inequalities \cite{Acin12}. Let $A_0, A_1, B_0, B_1$ be $\pm 1$-valued random variables. For a random variable $X$, let $\left<X\right>$ denote its expectation. The tilted CHSH inequality \cite{Acin12}, with parameter $\alpha \in [0,2]$, is the following generalization of the CHSH inequality:
\begin{equation}
    \left<\alpha A_0 + A_0B_0 + A_0B_1 +A_1B_0 - A_1B_1 \right > \leq 2+\alpha,
    \label{tiltedchsh}
\end{equation}
which holds when the random variables are local. The maximal quantum violation is $I_{\alpha} := \sqrt{8+2\alpha^2}$ and is attained when the strategy of the two parties consists of sharing the joint state $\ket{\psi} = \cos \theta \ket{00} + \sin \theta \ket{11}$, and measuring observables $A_0, A_1$ and $B_0, B_1$ respectively, where $A_0 = \sigma_z$, $A_1 = \sigma_x$, $B_0 = \cos \mu \sigma_Z + \sin \mu \sigma_X$ and $B_1 = \cos \mu \sigma_z + \sin \mu \sigma_x$, and $\sin 2\theta = \sqrt{\frac{4-\alpha^2}{4+\alpha^2}}$ (or $\alpha \equiv \alpha(\theta) = 2/\sqrt{1+2\tan^2{2\theta}}$) and $\mu = \arctan \sin 2\theta$. Here $\sigma_Z$ and $\sigma_X$ are the usual Pauli matrices. The converse also holds, in the sense that maximal violation self-tests this strategy. 



\section{The Bell inequality}
\label{sec: main}

The family of Bell inequalities that we are about to introduce is over question sets $\mathcal{X} = \{0,1,2\}$ and $\mathcal{Y} = \{0,1,2,3\}$, and answer sets $\mathcal{A} = \mathcal{B} = \{0,..,d-1\}$ (where $d \geq 2$ corresponds to the local dimension). We introduce some notation. For a correlation $p \in \mathcal{C}_q^{3,4,d,d}$ and $m \in  \{0,1,..\floor{\frac{d}{2}}-1\}$, define
\begin{equation}
\label{eq: chsh functional}
 [\mbox{CHSH}_m]_p := \sum_{x,y \in \{0,1\}, a,b \in \{2m,2m+1\}} (-1)^{a \oplus b - xy} p(a,b|x,y)   
\end{equation}
where $a \oplus b - xy$ is intended modulo $2$. Note that for $m=0$, this is the usual CHSH Bell functional. For $m>0$ the form is the same, but the answers are in $\{2m,2m+1\}$. In what follows, we will use the term ``standard CHSH'' to refer to the standard CHSH inequality or Bell functional on binary question and answer sets. This is to distinguish it from the new functionals we have just defined. We will also use the terms Bell operator and Bell functional interchangeably. 

We can define a similar functional to \eqref{eq: chsh functional} for questions $x \in \{0,2\}$ and $y \in \{2,3\}$ and answers in $\{2m+1, 2m+2\}$. Here questions $x \in \{0,2\}$ and $y \in \{2,3\}$ take the role of the $\{0,1\}$ questions in \eqref{eq: chsh functional}. So, for convenience of notation define a relabelling map $f: \{0,2\} \rightarrow \{0,1\}$ to be such that $f(0) = 0, f(2) = 1$, and a relabelling map $g: \{2,3\} \rightarrow \{0,1\}$ to be such that $g(2) = 0$, $g(3) = 1$. Then, define
\begin{equation*}
 [\mbox{CHSH}'_m]_p := \sum_{x \in \{0,2\}, y \in \{2,3\}, a,b \in \{2m+1,2m+2\}} (-1)^{a \oplus b - f(x)g(y)} p(a \mbox{ mod } d,b \mbox{ mod }d|x,y)   
\end{equation*}
From now onwards, we omit writing ``$\mbox{mod } d$'' for ease of notation, and the answers are intended $\mbox{mod } d$.

Denote by $\mathcal{C}$ and $\mathcal{C}'$ the sets
\begin{align}
\mathcal{C} &= \left \{(a,b,x,y): (x,y) \in \{0,1\} \times \{0,1\} \wedge (a,b) \notin \bigcup_{m=0}^{\floor{\frac{d}{2}}-1} \{2m, 2m+1\} \times \{2m,2m+1\} \right\}, \label{eq: C}\\
\mathcal{C}' &= \left \{(a,b,x,y):  (x,y) \in \{0,2\} \times \{2,3\} \wedge (a,b) \notin \bigcup_{m=0}^{\floor{\frac{d}{2}}-1} \{2m+1, 2m+2\} \times \{2m+1,2m+2\} \right\}. \label{eq: C prime}
\end{align}
Then, define the cross terms
\begin{align*}
[\mbox{CROSS}]_p &:= \sum_{a,b,x,y: (a,b,x,y) \in \mathcal{C}} p(a,b|x,y), \\
[\mbox{CROSS}']_p &:= \sum_{a,b,x,y: (a,b,x,y) \in \mathcal{C}'} p(a,b|x,y).
\end{align*}

We are ready to define the family of Bell operators for our inequalities.
\begin{definition}[The Bell operator]\label{def: bell-operators}
Let $d \geq 2 \in \mathbb{Z}$ and $\mathds{1}_{\{d>2\}}$ and $\mathds{1}_{\{d\,\,\small{ odd}\}}$ be the indicator functions for the cases $d>2$ and $d$ odd respectively. Let $\delta>0$ be a constant. For a correlation $p$, the Bell operator takes the form:
\begin{align}
[\B]_p &:= \sum_{m=0}^{\floor{\frac{d}{2}}-1} [\textnormal{CHSH}_m]_p +  \mathds{1}_{\{d>2\}} \cdot \sum_{m=0}^{\floor{\frac{d}{2}}-1} [\textnormal{CHSH}'_m]_p -\delta \cdot ([\textnormal{CROSS}]_p + [\textnormal{CROSS}']_p) \nonumber\\
&+ \mathds{1}_{\{d\,\,\small{ odd}\}} \cdot \frac{\sqrt{2}}{2} \cdot \left(\sum_{x,y \in \{0,1\}} p(d-1, d-1| x,y) + \sum_{x \in \{0,2\}, y \in \{2,3\}} p(0, 0 | x,y) \right) \label{eq: max bell operator}
\end{align}

\end{definition}

Intuitively the terms $\mbox{CROSS}$ and $\mbox{CROSS}'$ can be thought of as ``penalty'' terms: they are meant to enforce that any correlation maximizing the value of the Bell operator, must put zero probability mass on the cross terms from $\mathcal{C}$ and $\mathcal{C}'$. We will argue that it is enough to multiply these penalty terms by any arbitrarily small but positive constant $\delta$ to ensure that maximal violation is attained exclusively by the maximally entangled state. On the other hand, with a zero penalty, it is still the case that the corresponding Bell inequality can be maximally violated using a maximally entangled state, but we are unable to show that the self-testing result still holds true (i.e. the converse).


\begin{theorem}[Classical bound]\label{thm: classical-bound}
For any $d \geq 2$ and any $p \in \mathcal{C}_c^{3,4,d,d}$:
    $$[\B]_p \leq 2 \cdot (1+ \mathds{1}_{\{d>2\}})$$
\end{theorem}

\begin{proof}
For $d=2$ we recover the classical case of the standard CHSH inequality, so assume $d>2$ from now on. Finding the best classical strategy is equivalent to finding the best deterministic strategy. Let $f_A:\{0,1,2\}\rightarrow \{0,..,d-1\}$ and $f_B:\{0,1,2,3\}\rightarrow \{0,..,d-1\}$ be functions specifying a deterministic strategy. Now, suppose $f_A(0) \in \{2k,2k+1\}$, $f_A(1) \in \{2l,2l+1\}$ and $f_A(2) \in \{2l', 2l'+1\}$. 

\begin{itemize}
    \item If $k=l$, It's easy to see that the best choice for $f_B(0)$ and $f_B(1)$ is to have also $f_B(0), f_B(1) \in \{2k,2k+1\}$ and get a contribution of at most $2$ (this is from the standard CHSH classical bound)
    \item if $k \neq l$, it's also easy to see that the best choice for $f_B(0)$ and $f_B(1)$ is to have one of three possibilities: $f_B(0), f_B(1) \in \{2k, 2k+1\}$; $f_B(0), f_B(1) \in \{2l, 2l+1\}$; or one in $\{2k, 2k+1\}$ and the other in $\{2l,2l+1\}$. They all achieve a contribution of at most $2$.
\end{itemize}

Similarly, the best possible choice for $f_B(2)$ and $f_B(3)$ gives a contribution of $2$. This yields the desired bound.
\end{proof}

We turn to quantum correlations. We have the following two theorems:

\begin{theorem}[Quantum bound]\label{thm: quantum-bound}
For any $d$ even and any $p \in \mathcal{C}_q^{3,4,d,d}$:

\begin{equation}
[\B]_p \leq 2\sqrt{2} \cdot (1 + \mathds{1}_{\{d>2\}}).  \label{eq:7}
\end{equation}

\end{theorem}

\begin{theorem}[Exact self-testing]\label{thm: exact self-testing}
For any $d \geq 2$, there is a unique correlation which achieves the quantum bound of $\B$, and it self-tests the state $\ket{\Psi} = \frac{1}{\sqrt{d}}\sum_{i=0}^{d-1} \ket{ii}$.
\end{theorem}

\vspace{3mm}
\paragraph{Proof overview} At a high level, the proof of Theorems \ref{thm: quantum-bound} and \ref{thm: exact self-testing} goes through the following steps:
\begin{itemize}
    \item [(i)] The correlation from \cite{CGS17} (in the maximally entangled case), achieves the RHS of \eqref{eq:7} (Lemma \ref{lem: ideal corr});
    \item[(ii)] Any correlation achieving the maximal quantum value of the Bell operator must have zero probability mass on the cross terms. This is proved by starting from a correlation which achieves the maximum but has non-zero cross terms, and modifying this into a strategy for qubit CHSH which achieves a value strictly higher than $2\sqrt{2}$, which is a contradiction. (This is the content of Lemma \ref{lem: key});
    \item[(iii)] Having zero cross-terms forces the correlations to have the block-diagonal form of \cite{CGS17}. The $2 \times 2$ blocks are across pairs of answers $\{2m, 2m+1\}$ for questions $x,y \in \{0,1\}$ and across pairs of answers $\{2m+1, 2m+2\}$ for questions $x \in \{0,2\}, y \in \{2,3\}$ (Lemma \ref{lem: weights});
    \item[(iv)] Finally, the freedom in the value of the weights of the blocks is fixed by the requirement that the block-diagonal structure is both over pairs of answers $\{2m, 2m+1\}$, for $x,y \in \{0,1\}$, and also over pairs of answers $\{2m+1, 2m+2\}$, for $x \in \{0,2\}, y \in \{2,3\}$, and these two subsets of questions have in common the question $x = 0$.
\end{itemize}

We will now describe ideal correlations achieving the quantum bound of \eqref{eq:7}. For a single-qubit observable $A$, we denote by $(A)_m$ the observable defined with respect to the basis $\left(\ket{2m},\ket{2m+1}\right)$. For example, $(\sigma_Z)_m = \ket{2m}\bra{2m} - \ket{2m+1}\bra{2m+1}$. Similarly, we denote by $(A)'_m$ the observable defined with respect to the basis $\left(\ket{2m+1}, \ket{2m+2}\right)$.

\begin{lem}
[Ideal correlation from \cite{CGS17} achieving the quantum bound]
\label{lem: ideal corr}
The correlation $p^* \in \mathcal{C}_q^{3,4,d,d}$ specified by the following quantum strategy $(\ket{\Psi}, \{\Pi_{A_x}^a\}_a, \{\Pi_{B_y}^b\}_b\})$ achieves the RHS of \eqref{eq:7}:
\begin{itemize}
    \item $\ket{\Psi} = \frac{1}{\sqrt{d}}\sum_{i=0}^{d-1} \ket{ii}$
    \item For $m = 0,.., \floor{\frac{d}{2}} -1$:
    \begin{itemize}
        \item $\Pi_{A_0}^{2m}$, $\Pi_{A_0}^{2m+1}$ are the projectors respectively onto the $+1,-1$ eigenspaces of $(\sigma_Z)_m$. (in other words, the measurement for $x=0$ is in the computational basis);
        \item $\Pi_{A_1}^{2m}$, $\Pi_{A_1}^{2m+1}$ onto the $+1,-1$ eigenspaces of $(\sigma_X)_m$. If $d$ is odd, $\Pi_{A_1}^{d-1} = \ket{d-1}\bra{d-1}$
        \item $\Pi_{A_2}^{2m+1}$, $\Pi_{A_2}^{2m+2}$ onto the $+1,-1$ eigenspaces of $(\sigma_X)'_m$. If $d$ is odd, $\Pi_{A_2}^{0} = \ket{0}\bra{0}$.

    \end{itemize} 
    \item For $m = 0,.., \floor{\frac{d}{2}} -1$:
    \begin{itemize}
        \item For $y \in \{0,1\}$, $\Pi_{B_y}^{2m}$, $\Pi_{B_y}^{2m+1}$ are the projectors respectively onto the $+1,-1$ eigenspaces of $(\frac{\sigma_Z + (-1)^{y} \sigma_X}{\sqrt{2}})_m$. If $d$ is odd, $\Pi_{B_y}^{d-1} = \ket{d-1}\bra{d-1}$; 
        \item For $y \in \{2,3\}$, $\Pi_{B_y}^{2m+1}$, $\Pi_{B_y}^{2m+2}$ onto the $+1,-1$ eigenspaces of $(\frac{\sigma_Z + (-1)^{y} \sigma_X}{\sqrt{2}})'_m$. If $d$ is odd, $\Pi_{B_y}^{0} = \ket{0}\bra{0}$. 
    \end{itemize}
\end{itemize}
\end{lem}

\begin{proof}
This is a straightforward check.
\end{proof}

\begin{lem}[Zero mass on the cross terms]\label{lem: key}
Let $p \in \mathcal{C}_q^{3,4,d,d}$ be a quantum correlation achieving maximal quantum value of $\B$. Then, $p(a,b|x,y) = 0$ $\forall (a,b,x,y) \in \mathcal{C} \cup \mathcal{C}'$, where $\mathcal{C}$ and $\mathcal{C}'$ are as in equations \eqref{eq: C} and \eqref{eq: C prime}.
\end{lem}
This establishes that any correlation maximally violating the Bell inequality must have the same block-diagonal form of the self-testing correlation from Lemma \ref{lem: ideal corr}.
\begin{proof}
We argue first for the case of $d$ even. 
We will show that any correlation achieving maximal value of $\B$ must have $p(a,b|x,y) = 0 \,\,\, \forall (a,b,x,y) \in \mathcal{C} \cup \mathcal{C}'$.
Suppose for a contradiction that a correlation $p \in \mathcal{C}_q^{3,4,d,d}$ achieves the maximal value of $\B$ and $p(a,b|x,y) = \gamma > 0$ for some $(a,b,x,y) \in \mathcal{C} \cup \mathcal{C}'$. 
In order to compensate for the negative contribution due to the presence of the cross terms in \eqref{eq: max bell operator} (which are multiplied by an arbitrary small but positive constant $\delta$), it must be the case that either $\sum_{m=0}^{\frac{d}{2}-1} [\textnormal{CHSH}_m]_p > 2\sqrt{2}$ or $\sum_{m=0}^{\frac{d}{2}-1} [\textnormal{CHSH}'_m]_p > 2\sqrt{2}$ (since we know from Lemma \ref{lem: ideal corr} that the maximal value of $\mathcal{B}$ is at least $2 \cdot 2\sqrt{2}$.). Assume the former (the other case being similar).

Let $S = (\ket{\psi}, \Pi_{A_x}^a, \Pi_{B_y}^b)$ be a quantum strategy producing correlation $p$. We will use this to construct a correlation $\tilde{p} \in \mathcal{C}_q^{2,2,2,2}$ that achieves a value of CHSH greater than $2\sqrt{2}$, which would be a contradiction. This is achieved by starting from strategy $S$ and mapping each pair of answers $(2k,2k+1)$ in $\{2,..,d-1\}$ to either their parity or the opposite of their parity, i.e either $(2k,2k+1) \mapsto (0,1)$ or $(2k,2k+1) \mapsto (1,0)$. More precisely, for $\vec{o} \in \{0,1\}^{\frac{d}{2}-1}$ let $\vec{o}[m]$ denote the $m$th bit of $\vec{o}$, and define a new quantum strategy for standard CHSH $S^{(\vec{o})} = (\ket{\psi}, \{\tilde{\Pi}_{A_x}^{a}\}_{a,x \in \{0,1\}}, \{\tilde{\Pi}_{B_y}^{b}\}_{b,y \in \{0,1\}})$ on the same state $\ket{\psi}$, with projectors, for $x,y \in \{0,1\}$,
\begin{align}
    \tilde{\Pi}_{A_x}^{0} &= \Pi_{A_x}^{0} +  \sum_{m=1}^{\frac{d}{2}-1}\Pi_{A_x}^{2m + \vec{o}[m]} \,\,\,\,\, &\tilde{\Pi}_{A_x}^{1} = \Pi_{A_x}^{1} +  \sum_{m=1}^{\frac{d}{2}-1}\Pi_{A_x}^{2m + 1 -\vec{o}[m]} \\
    \tilde{\Pi}_{B_y}^{0} &= \Pi_{B_y}^{0} +  \sum_{m=1}^{\frac{d}{2}-1}\Pi_{B_y}^{2m + \vec{o}[m]} \,\,\,\,\, &\tilde{\Pi}_{B_y}^{1} = \Pi_{B_y}^{1} +  \sum_{m=1}^{\frac{d}{2}-1}\Pi_{B_y}^{2m + 1 -\vec{o}[m]}
\end{align}
Let $\tilde{p}^{(\vec{o})}$ be the resulting correlation. Now, let $[\textnormal{CHSH}]_{\tilde{p}^{(\vec{o})}}$ be the CHSH value of correlation $\tilde{p}^{(\vec{o})}$. Since CHSH is an XOR game (i.e. only the xor of the answers matters), it's easy to see that for any $\vec{o} \in\{0,1\}^{\frac{d}{2}-1}$ 

\begin{equation}
[\textnormal{CHSH}]_{\tilde{p}^{(\vec{o})}} = \sum_{m=0}^{\frac{d}{2}-1} [\textnormal{CHSH}_m]_p + C
\end{equation}
where $C$ is a (possibly negative) contribution which comes from the cross terms of the form $\bra{\psi} \Pi_{A_x}^{a} \otimes \Pi_{B_y}^{b} \ket{\psi}$ for $(a,b,x,y) \in \mathcal{C}$. However, there exists a choice of $\vec{o} \in \{0,1\}^{\frac{d}{2}-1}$ such that $C \geq 0$. In fact, notice that the contributions to $C$ coming from cross terms involving $(2m,2m+1)$ when one chooses $\vec{o}[m] = 0$ or $\vec{o}[m] = 1$ (and keeps the other choices fixed) are the negative of each other. Hence at least one of the two choices gives a non-negative contribution. Then, pick $\vec{o} \in \{0,1\}^{\frac{d}{2}-1}$ as follows: for $m = 1,.., \frac{d}{2}-1$, in this order, choose a value of $\vec{o}[m]$ for which the contribution from cross terms involving pairs $(2m,2m+1)$ and $(2m',2m'+1)$ for $m'<m$ is non-negative.  
This gives $C \geq 0$.

So, for this choice of $\vec{o}$, one gets $[\textnormal{CHSH}]_{\tilde{p}^{(\vec{o})}} > 2\sqrt{2}$, which is the desired contradiction.

The case of $d$ odd is similar but requires slightly more effort. Suppose $p \in \mathcal{C}_q^{3,4,d,d}$ achieves the maximal value of $\B$ and $p(a,b|x,y) = \gamma > 0$ for some $(a,b,x,y) \in \mathcal{C} \cup \mathcal{C}'$. Then it must be the case that either $\sum_{m=0}^{\frac{d}{2}-1} [\textnormal{CHSH}_m]_p + \frac{\sqrt{2}}{2}\cdot\sum_{x,y \in \{0,1\}} p(d-1, d-1 | x,y) > 2\sqrt{2}$ or $\sum_{m=0}^{\frac{d}{2}-1} [\textnormal{CHSH}'_m]_p+ \frac{\sqrt{2}}{2}\cdot\sum_{x \in \{0,2\}, y \in \{2,3\}} p(0, 0 | x,y) > 2\sqrt{2}$. Suppose the former (the latter case being similar). Let $S = (\ket{\psi}, \Pi_{A_x}^a, \Pi_{B_y}^b)$ be a quantum strategy producing correlation $p$. For a string $\vec{o} \in \{0,1\}^{\frac{d}{2}-1}$, we construct the following strategy for CHSH $S^{(\vec{o})} = (\ket{\tilde{\psi}}, \{\tilde{\Pi}_{A_x}^{a}\}_{a,x \in \{0,1\}}, \{\tilde{\Pi}_{B_y}^{b}\}_{b,y \in \{0,1\}})$: intuitively, the two parties share the original state tensored with an EPR pair. They map outcomes $\{0,..,d-2\}$ to outcomes in $\{0,1\}$ (similarly as before). If one sees outcome $d-1$, they measure the shared EPR pair with an appropriate ideal CHSH measurement. More precisely, let $\{P_{A_x}^{a}\}_{a,x \in \{0,1\}}, \{P_{B_y}^{b}\}_{b,y \in \{0,1\}}$ be the ideal CHSH qubit measurements. Then, $\ket{\tilde{\psi}} = \ket{\psi} \otimes \ket{\textnormal{EPR}}$, and
\begin{align*}
        \tilde{\Pi}_{A_x}^{0} &= [\Pi_{A_x}^{0} +  \sum_{m=1}^{\floor{\frac{d}{2}}-1}\Pi_{A_x}^{2m + \vec{o}[m]}] \otimes I + \Pi_{A_x}^{d-1} \otimes P_{A_x}^{0}, \,\,\,
        &\tilde{\Pi}_{A_x}^{1} &= [\Pi_{A_x}^{1} +  \sum_{m=1}^{\floor{\frac{d}{2}}-1}\Pi_{A_x}^{2m +1 - \vec{o}[m]}] \otimes I + \Pi_{A_x}^{d-1} \otimes P_{A_x}^{1} \\
        \tilde{\Pi}_{B_y}^{0} &= [\Pi_{B_y}^{0} +  \sum_{m=1}^{\floor{\frac{d}{2}}-1}\Pi_{B_y}^{2m + \vec{o}[m]}] \otimes I + \Pi_{B_y}^{d-1} \otimes P_{B_y}^{0}, \,\,\, 
        &\tilde{\Pi}_{B_y}^{1} &= [\Pi_{B_y}^{1} +  \sum_{m=1}^{\floor{\frac{d}{2}}-1}\Pi_{B_y}^{2m + 1 - \vec{o}[m]}] \otimes I + \Pi_{B_y}^{d-1} \otimes P_{B_y}^{1}
    \end{align*}
One can check, then, that with the appropriate choice of $\vec{o}$ (chosen similarly to the $d$ even case), this gives a strategy for CHSH which achieves a value strictly greater than $2\sqrt{2}$.

\end{proof}

The following lemma establishes that if a correlation $p$ has zero cross-terms, then this implies that the restriction of $p$ to the subset of questions $(x,y) \in \{0,1\}^2$ and to answers $a,b \in \{2m, 2m+1\}$ is still a correlation (multiplied by some weight). Likewise for the restriction to the subset of questions $(x,y) \in \{0,2\} \times \{2,3\}$ and to answers $a,b \in \{2m+1, 2m+2\}$. 
\begin{lem}\label{lem: weights}
Any correlation $p \in \mathcal{C}_q^{3,4,d,d}$ with zero cross-terms (i.e of the form of Lemma \ref{lem: key}), induced by some strategy $\left(\ket{\psi}, \{\Pi_{A_x}^a\}_a, \{\Pi_{B_y}^b\}_b \right)$, satisfies the following:

\begin{itemize}
    \item If $d$ is even, for each $m=0,..,\frac{d}{2}-1$, there exist weights $w_m, w_m' \geq 0$ with $\sum_m w_m = 1$, $\sum_m w'_m = 1$ and correlations $p_m, p_m' \in \mathcal{C}_q^{2,2,2,2}$ (with questions in $\{0,1\}^2$ and $\{0,2\}\times \{2,3\}$ respectively, and answers in $\{0,1\}$) such that $\forall m$, $\forall a,b \in \{2m,2m+1\}, x,y \in \{0,1\}$: 
    \begin{equation}
    \label{eq: 10}
        p(a,b|x,y) = w_m \cdot p_m(a\mod 2,b\mod 2|x,y) = \bra{\psi}\Pi_{A_x}^a\otimes \Pi_{B_y}^b \ket{\psi}
    \end{equation}

and $\forall m$, $\forall a,b \in \{2m+1,2m+2\}, x \in \{0,2\}, y \in \{2,3\}$: \begin{equation}
\label{eq: 11}
    p(a,b|x,y) = w_m' \cdot p_m'(a \mod 2,b \mod 2|x,y) = \bra{\psi}\Pi_{A_x}^a\otimes \Pi_{B_y}^b \ket{\psi}
\end{equation}

\item If $d$ is odd, the analogous statement holds, except that the weights $w_m, w_m'$ are such that $\sum_m w_m + p(d-1,d-1|0,0)  = \sum_m w'_m + p(0,0|2,2) = 1$, AND 
\begin{itemize}
    \item $p(d-1,d-1|x,y) = p(d-1,d-1|x',y') \,\, \forall x,y,x'y' \in \{0,1\}$
    \item $p(0,0|x,y) = p(0,0|x',y') \,\, \forall x,x' \in \{0,2\}, y,y' \in \{2,3\}$
\end{itemize}

\end{itemize}

\end{lem}

\begin{proof}
Let $p \in \mathcal{C}_q^{3,4,d,d}$ be of the form of Lemma \ref{lem: key}, and let $(\ket{\psi}$, $\{\Pi_{A_x}^{a}\}$, $\{\Pi_{B_y}^b\})$, be a strategy reproducing $p$. Then, for $m=0,..,\frac{d}{2}-1$ define:
\begin{itemize}
    \item[(i)] for $x,y \in \{0,1\}$, $A_x^{(m)} = \Pi_{{A}_x}^{2m} - \Pi_{{A}_x}^{2m+1}$ and $B_y^{(m)} = \Pi_{B_y}^{2m} - \Pi_{{B}_y}^{2m+1}$ 
    \item[(ii)] for $x \in \{0,2\}, y \in \{2,3\}$, $A_x^{\prime (m)} = \Pi_{{A}_x}^{2m+1} - \Pi_{{A}_x}^{2m+2}$ and $B_y^{\prime (m)}= \Pi_{{B}_y}^{2m+1} - \Pi_{{B}_y}^{2m+2}$
\end{itemize}

Define the subspaces $\mathcal{U}_m = \text{Range}(A_0^{(m)})+ \text{Range}(A_1^{(m)})$ and $\mathcal{V}_m = \text{Range}(B_0^{(m)})+ \text{Range}(B_1^{(m)})$, and let $\mathds{1}_{{\mathcal{U}}_m}$ and $\mathds{1}_{{\mathcal{V}}_m}$ be projections onto these subspaces. Let $\ket{\psi_m} := \mathds{1}_{{\mathcal{U}}_m}\mathds{1}_{{\mathcal{V}}_m} \ket{\psi}$

We will check that $\mathds{1}_{{\mathcal{U}}_m} \ket{\psi} = \mathds{1}_{\text{Range}(A_0^{(m)})}\ket{\psi} =  \mathds{1}_{\text{Range}(B_0^{(m)})} \ket{\psi} = \mathds{1}_{{\mathcal{V}}_m} \ket{\psi} = \ket{\psi_m}$. We compute 
\begin{align}
\mathds{1}_{\text{Range}(A_0^{(m)})}\ket{\psi} &= \left(\Pi_{A_0}^{2m}+ \Pi_{A_0}^{2m+1}\right) \ket{\psi} \\
&= \left(\Pi_{A_0}^{2m}+ \Pi_{A_0}^{2m+1}\right) \sum_{l=0}^{d-1} \Pi_{B_0}^{l}\ket{\psi} \\
&= \left(\Pi_{A_0}^{2m}+ \Pi_{A_0}^{2m+1}\right) \left(\Pi_{B_0}^{2m}+ \Pi_{B_0}^{2m+1}\right) \ket{\psi} \label{eq: 12} \\
&= \mathds{1}_{\text{Range}(A_0^{(m)})} \mathds{1}_{\text{Range}(B_0^{(m)})} \ket{\psi}, \label{eq: 14}
\end{align}
where the third line follows from the hypothesis that the correlation has the form of Lemma \ref{lem: key}. The same calculation starting from $\mathds{1}_{\text{Range}(B_0^{(m)})} \ket{\psi}$ gives $\mathds{1}_{\text{Range}(B_0^{(m)})} \ket{\psi}= \mathds{1}_{\text{Range}(A_0^{(m)})} \mathds{1}_{\text{Range}(B_0^{(m)})} \ket{\psi}$, which, together with \eqref{eq: 14}, implies $\mathds{1}_{\text{Range}(A_0^{(m)})} \ket{\psi} = \mathds{1}_{\text{Range}(B_0^{(m)})} \ket{\psi}$. With similar calculations, we also deduce $\mathds{1}_{\text{Range}(A_1^{(m)})} \ket{\psi} = \mathds{1}_{\text{Range}(B_0^{(m)})} \ket{\psi}$, which implies $\mathds{1}_{\text{Range}(A_0^{(m)})}  \ket{\psi} = \mathds{1}_{\text{Range}(A_1^{(m)})} \ket{\psi}$, and hence $\mathds{1}_{{\mathcal{U}}_m} \ket{\psi} = \mathds{1}_{\text{Range}(A_0^{(m)})}\ket{\psi}$. Similarly $\mathds{1}_{\text{Range}(B_0^{(m)})}  \ket{\psi} = \mathds{1}_{\text{Range}(B_1^{(m)})} \ket{\psi}$, and hence $\mathds{1}_{{\mathcal{V}}_m} \ket{\psi} = \mathds{1}_{\text{Range}(B_0^{(m)})}\ket{\psi}$. 
Altogether, we have deduced that \begin{equation}
\label{eq: 14}
    \mathds{1}_{{\mathcal{U}}_m} \ket{\psi} = \mathds{1}_{\text{Range}(A_0^{(m)})}\ket{\psi} =  \mathds{1}_{\text{Range}(B_0^{(m)})} \ket{\psi} = \mathds{1}_{{\mathcal{V}}_m} \ket{\psi} = \ket{\psi_m}
\end{equation}.

Finally, set $w_m = \|\ket{\psi_m}\|^2$ to get the desired weights, and take the correlations $p_m$ as in \eqref{eq: 10}. We argue similarly for the weights $w'_m$ and the correlations $p_m'$. A very similar argument yields the conclusion for the case of odd $d$.

\end{proof}

\begin{cor}
\label{cor: 1}
Any correlation $p \in \mathcal{C}_q^{3,4,d,d}$ with zero cross-terms  (i.e. of the form of Lemma \ref{lem: key}) satisfies the following:

\begin{itemize}
\item If $d$ is even, there exist weights $w_m, w'_m \geq 0$, $m=0,..,\frac{d}{2}-1$, with $\sum_m w_m = 1$, $\sum_m w'_m = 1$, such that, for all $m$,
$$[\textnormal{CHSH}_m]_p \leq w_m \cdot 2\sqrt{2}$$ and 
$$[\textnormal{CHSH}'_m]_p \leq w'_m \cdot 2\sqrt{2}$$

\item If $d$ is odd, the analogous statement holds, except that the weights $w_m, w_m'$ are such that $\sum_m w_m + p(d,d|0,0) = 1$, $\sum_m w'_m + p(0,0|2,2) = 1$.
\end{itemize}
\end{cor}

\begin{proof}
This follows immediately from Lemma \ref{lem: weights}.
\end{proof}

\begin{proof}[Proof of Theorems \ref{thm: quantum-bound} and \ref{thm: exact self-testing}]
Assume $d>2$, as the $d=2$ case corresponds to standard CHSH. We start with $d$ even (the odd case being similar). Let $p \in \mathcal{C}_q^{3,4,d,d}$ be a correlation that achieves the maximal quantum value of $\B$. By Lemma \ref{lem: key}, $p$ must have zero cross-terms. Then, from Lemma \ref{lem: weights}, we deduce, for $m = 0,..,\frac{d}{2}-1$, the existence of weights $w_m, w_m'$ and correlations $p_m, p_m'$ satisfying the statement of the Lemma. This implies
\begin{equation}
\sum_{m=0}^{\frac{d}{2}-1} [\textnormal{CHSH}_m]_p = \sum_{m=0}^{\frac{d}{2}-1} w_m \cdot [\textnormal{CHSH}]_{p_m} \leq 2\sqrt{2} \label{eq:9}
\end{equation}
where we have bounded each term with the standard CHSH bound. 
Similarly, we also get $\sum_{m=0}^{\frac{d}{2}-1} [\textnormal{CHSH}^{\prime}_m]_p\leq 2\sqrt{2}$, which implies the desired upper bound of Theorem \ref{thm: quantum-bound}.

Such upper bound is achieved if and only if $[\textnormal{CHSH}]_{p_m}
= w_m \cdot 2\sqrt{2}$ for all $m$, and $[\textnormal{CHSH}^{\prime}_m]_p
= w'_m \cdot 2\sqrt{2}$ for all $m$. This is if and only if:
\begin{itemize}
    \item for all $m$,
    $w_m = 0$ OR $p_m$ is the ideal qubit CHSH correlation, AND
    \item for all $m$, $w'_m = 0$ OR $p_m'$ is the ideal qubit CHSH correlation
\end{itemize}

We want to argue that the only way that this can happen is if the weights are all equal (and non-zero). Once we have shown this, we notice that we have specified the correlation $p$ completely for the two subsets of questions $x,y \in \{0,1\}$ and $x\in \{0,2\}, y \in \{2,3\}$. From \cite{CGS17}, we know this is enough to uniquely determine the self-testing correlation for the maximally entangled state of local dimension $d$ presented in \cite{CGS17} (and in Lemma \ref{lem: ideal corr}), and we thus deduce that maximal violation of the Bell inequality self-tests $\ket{\Psi}$.

Let ($\ket{\psi}$, $\{\Pi_{A_x}^a\}_a$, $\{\Pi_{B_y}^b\}_b$ be a quantum strategy for $p$ (which achieves the upper bound). Then, by what we have argued above, for all $m$ we have $\|\Pi_{A_0}^{2m+1} \ket{\psi}\|^2 = w_{m} \cdot \frac12$, and this holds both when $w_m  \neq 0$ (and $p_m$ is the ideal qubit CHSH correlation) and when $w_m = 0$.
Likewise, we have that $\|\Pi_{A_0}^{2m+1} \ket{\psi}\|^2 = w'_m \cdot \frac12$. And similarly $\|\Pi_{A_0}^{2m} \ket{\psi}\|^2 = w_{m} \cdot \frac12$ and $\|\Pi_{A_0}^{2m} \ket{\psi}\|^2 = w'_{m-1} \cdot \frac12$. Clearly this, together with the constraint $\sum_{m}w_m = \sum_{m} w_m' = 1$, implies $w_m = w_m' = \frac{2}{d} \,\, \,\forall m$.

The proof is similar for the case of $d$ odd, where we instead deduce $w_m = w_m' = \frac{2}{d} \,\,\forall m$ (there are $\frac{d-1}{2}$ values of $m$) and $ p(d-1,d-1|x,y)= p(0,0|x',y') =  \frac{1}{d} \,\, \forall x,y \in \{0,1\}, x'\in \{0,2\}, y' \in \{2,3\}$.
\end{proof}

A robust version of the self-testing result \textit{via the correlations} of \cite{CGS17} was shown in \cite{CS17}, where, informally, the authors prove that a strategy producing a correlation that is $\epsilon$-close to the ideal one, must be $O(d^3\epsilon^{\frac14})$-close (according to some measures of distance) to the ideal strategy from Lemma \ref{lem: ideal corr}. However, this does not trivially translate to a robust self-test \textit{via our Bell inequality}, for which we require that a \textit{close-to-maximal violation} certifies a close-to-ideal strategy. Since translating the exact analysis to a robust analysis is not particularly illuminating, we leave the details to the appendix. For the robust self-testing theorem via our Bell inequality, refer to Theorem \ref{thm: robust self test} in the Appendix. 

\section{Generalizing the tilted CHSH inequalities (a conjecture)}
\label{sec: conj}
Let $I_{\alpha} = \sqrt{8+2\alpha^2}$ be the maximal quantum violation of the tilted CHSH inequality, for coefficient $\alpha$. The family of candidate Bell inequalities which we will describe is a very natural generalization of the Bell inequality from the previous section to the tilted case. We introduce some notation. For a correlation $p \in \mathcal{C}_q^{3,4,d,d}$, define
\begin{equation}
[\mbox{tCHSH}_m(\alpha)]_p := \,\, \alpha [p(a=2m|x=0) - p(a=2m+1|x=0)] + [\mbox{CHSH}_m]_p
\end{equation}
where $[\mbox{CHSH}_m]_p$ was defined earlier.
This can be thought of as a tilted CHSH Bell operator restricted to answers in $\{2m,2m+1\}$. Note that the above involves only questions $x,y \in \{0,1\}$. We can define a similar term for questions in $x \in \{0,2\}$ and $y \in \{2,3\}$ and answers in $\{2m+1, 2m+2\}$. Let
\begin{equation}
[\mbox{tCHSH}'_m(\alpha)]_p := \,\, \alpha [p(a=2m+1|x=0) - p(a=2m+2|x=0)] + [\mbox{CHSH}'_m]_p
\end{equation}

The sets $\mathcal{C}$ and $\mathcal{C}'$ of questions and answers corresponding to cross terms are defined as in the previous section. Then our candidate family of Bell operators generalizing the family of tilted CHSH inequalities is the following:

\begin{definition}[The family of Bell operators]\label{def: bell-operators-tilted}
Each inequality in the family is specified by:
\begin{itemize}
\item[(i)] $0 < c_i  < 1 \in \mathbb{R}$, $i=0,..,d-1$, with $\sum_{i=0}^{d-1} c_i^2 = 1$,
\item[(ii)] $d\geq 2 \in \mathbb{N}$
\end{itemize}
Let $\theta_m = \arctan{\frac{c_{2m+1}}{c_{2m}}}$, $\alpha_m \equiv \alpha_m(\theta_m) \in [0,2)$ be defined by $\sin 2\theta_m = \sqrt{\frac{4-\alpha_m^2}{4+\alpha_m^2}}$, $\theta'_m = \arctan{\frac{c_{2m+2}}{c_{2m+1}}}$, $\alpha'_m \equiv \alpha'_m(\theta'_m) \in [0,2)$ defined by $\sin 2\theta'_m = \sqrt{\frac{4-\alpha_m^{\prime 2}}{4+\alpha_m^{\prime 2}}}$. Let $\delta>0$ be a constant. For a correlation $p \in \mathcal{C}_q^{3,4,d,d}$, the Bell operator takes the form:
\begin{align}
\label{eq: tilted Bell operator}
[\tB(c_0,..,c_{d-1})]_p :=&\,\,   \sum_{m=0}^{\floor{\frac{d}{2}}-1} \frac{1}{I_{\alpha_m}}  [\textnormal{tCHSH}_m(\alpha_m)]_p  + \mathds{1}_{\{d > 2\}}\cdot \sum_{m=0}^{\floor{\frac{d}{2}}-1} \frac{1}{I_{\alpha'_m}}[\textnormal{tCHSH}'_m(\alpha'_m)]_p - \delta \cdot ([\textnormal{CROSS}]_p + [\textnormal{CROSS}']_p)
 \nonumber \\ 
 &+ \mathds{1}_{\{d\,\,\small{ odd}\}} \cdot \frac{1}{4} \cdot \left(\sum_{x,y \in \{0,1\}} p(d-1, d-1| x,y) + \sum_{x \in \{0,2\}, y \in \{2,3\}} p(0, 0 | x,y) \right)
\end{align}
\end{definition}

Note that to put the Bell operator for the maximally entangled case in this form one just needs to divide \eqref{eq: max bell operator} by $2\sqrt{2}$.

\begin{conjecture}[Quantum bound and self-testing]
For any $d$ even and any $p \in \mathcal{C}_q^{3,4,d,d}$:
$$[\tB(c_0,..,c_{d-1})]_p \leq 1 + \mathds{1}_{\{d>2\}}$$
Moreover, there is a unique quantum correlation achieving the bound, and it self-tests the state $\ket{\Psi} = \sum_{i=0}^{d-1} c_i \ket{ii}$.
\end{conjecture}

The lack of symmetry in the tilted case seems to make the analysis surprisingly less straightforward, and the arguments we employed in the maximally entangled case do not directly carry over. 

An open question that applies to both the maximally entangled and the tilted Bell operators is to determine if cross terms are necessary for the self-testing property to hold true (i.e. whether, in \eqref{eq: max bell operator} and \eqref{eq: tilted Bell operator},  $\delta>0$ is necessary or $\delta=0$ suffices). 

\section{Acknowledgements}
The author thanks Koon Tong Goh and Thomas Vidick for helpful discussions, and thanks the latter for useful comments on an earlier version of this work. The author appreciates support from the Kortschak Scholars program, and AFOSR YIP award number FA9550-16-1-0495.

\bibliographystyle{alpha}
\bibliography{references}

\begin{appendix}
\section{Appendix}
\label{sec: appendix}
\subsection{Robustness}

Obtaining a robust self-testing result is mainly a matter of going through the proof and replacing exact statements with approximate statements, where necessary. Here, we first state the robust self-testing theorem, then we give an outline of the proof pointing out the parts where it differs from the proof for the exact case.

\begin{theorem}[Robust self-testing]
\label{thm: robust self test}
Let $\mathcal{B}$ be the Bell operator from Definition $\ref{def: bell-operators}$ with parameters $d \geq 2, \delta >0$. Let $\left(\ket{\Psi}, \{\tilde{\Pi}_{A_x}^{a}\}_{a}, \{\tilde{\Pi}_{B_y}^{b}\}_{b} \right)$ be the ideal strategy from Lemma \ref{lem: ideal corr}, where $\ket{\Psi} = \frac{1}{\sqrt{d}}\sum_{i=0}^{d-1} \ket{ii}$. There exists a constant $C>0$ such that the following holds. Suppose the strategy $\left(\ket{\psi}, \{\Pi_{A_x}^a\}_a, \{\Pi_{B_y}^b\}_b \right)$ attains a correlation $p$ such that $[\mathcal{B}]_p > 2\sqrt{2} -\epsilon$, for some $\epsilon < \frac{C}{d^3}$. Then, there exists a local unitary $\Phi$ and an auxiliary state $\ket{aux}$ such that 
\begin{align}
    \| \Phi(\ket{\psi}) -  \ket{\Psi} \otimes \ket{aux}  \| &= O(d^6\epsilon^{\frac18}) \\
     \| \Phi(\Pi_{A_x}^a \otimes \Pi_{B_y}^b \ket{\psi}) -  \tilde{\Pi}_{A_x}^a \otimes \tilde{\Pi}_{B_y}^b \ket{\Psi} \otimes \ket{aux}  \| &= O(d^6\epsilon^{\frac18}).
\end{align}
\end{theorem}

In the rest of this section we sketch the proof of Theorem \ref{thm: robust self test}. In doing so, we will state approximate versions of Lemma \ref{lem: key} and \ref{lem: weights} from the main text.

\begin{lem}[Approximate version of Lemma \ref{lem: key}]\label{lem: key approx}
Let $\mathcal{B}$ be the Bell operator with parameters $d \geq 2$ and $\delta >0$. Let $p \in \mathcal{C}_q^{3,4,d,d}$ be a quantum correlation such that $[\mathcal{B}]_p > 2\sqrt{2} -\epsilon$. Then, $p(a,b|x,y) <\frac{2}{\delta} \epsilon$ for all $(a,b,x,y) \in \mathcal{C} \cup \mathcal{C}'$, where $\mathcal{C}$ and $\mathcal{C}'$ are as in equations \eqref{eq: C} and \eqref{eq: C prime}.
\end{lem}
\begin{proof}
The proof is very similar to the proof of Lemma \ref{lem: key}. The only difference is that we now suppose for a contradiction that $p$ is such that $[\mathcal{B}]_p > 2\sqrt{2} -\epsilon$ and $p(a,b|x,y) \geq \frac{2}{\delta} \epsilon$ for some $(a,b,x,y) \in \mathcal{C} \cup \mathcal{C}'$. Then in order to compensate for a negative contribution $\geq \delta \cdot \frac{2}{\delta} \epsilon = 2\epsilon$, it must be that either  $\sum_{m=0}^{\frac{d}{2}-1} [\textnormal{CHSH}_m]_p > 2\sqrt{2}$ or $\sum_{m=0}^{\frac{d}{2}-1} [\textnormal{CHSH}'_m]_p > 2\sqrt{2}$. In either case, analogously to the proof of Lemma \ref{lem: key}, one can reduce this to a strategy that wins CHSH with value $>2\sqrt{2}$.
\end{proof}

\begin{lem}[Approximate version of Lemma \ref{lem: weights}]
\label{lem: weights approx}
Any correlation $p \in \mathcal{C}_q^{3,4,d,d}$ such that each cross term has size $O(\epsilon)$ (i.e. of the form of Lemma \ref{lem: key approx} - we are thinking of $\delta$ as a constant), induced by some strategy $\left(\ket{\psi}, \{\Pi_{A_x}^a\}_a, \{\Pi_{B_y}^b\}_b \right)$, satisfies the following:

\begin{itemize}
    \item If $d$ is even, then for each $m=0,..,\frac{d}{2}-1$, there exist weights $w_m, w_m' \geq 0$ with $1- O(\epsilon) \leq \sum_m w_m, \sum_m w'_m  \leq 1$, and correlations $p_m, p_m' \in \mathcal{C}_q^{2,2,2,2}$ (with questions in $\{0,1\}^2$ and $\{0,2\}\times \{2,3\}$ respectively, and answers in $\{0,1\}$) such that $\forall m$, $\forall a,b \in \{2m,2m+1\}, x,y \in \{0,1\}$: $$p(a,b|x,y) \approx_{O(\epsilon)} w_m \cdot p_m(a\mod 2,b\mod 2|x,y) \approx_{O(\epsilon)} \bra{\psi}\Pi_{A_x}^a\otimes \Pi_{B_y}^b \ket{\psi}$$

and $\forall m$, $\forall a,b \in \{2m+1,2m+2\}, x \in \{0,2\}, y \in \{2,3\}$: $$p(a,b|x,y) \approx_{O(\epsilon)} w_m' \cdot p_m'(a \mod 2,b \mod 2|x,y) \approx_{O(\epsilon)} \bra{\psi}\Pi_{A_x}^a\otimes \Pi_{B_y}^b \ket{\psi}$$.

\item If $d$ is odd, the analogous statement holds, except that the weights $w_m, w_m'$ are such that $1-O(\epsilon) \leq \sum_m w_m + p(d-1,d-1|0,0), \sum_m w'_m + p(0,0|2,2) \leq 1 + O(\epsilon)$, AND 
\begin{itemize}
    \item $p(d-1,d-1|x,y) \approx_{O(\epsilon)} p(d-1,d-1|x',y') \,\, \forall x,y,x'y' \in \{0,1\}$
    \item $p(0,0|x,y) \approx_{O(\epsilon)} p(0,0|x',y') \,\, \forall x,x' \in \{0,2\}, y,y' \in \{2,3\}$
\end{itemize}

\end{itemize}
\end{lem}

\begin{proof}
All equalities from \eqref{eq: 12} to \eqref{eq: 14} now hold approximately, up to addition of orthogonal vectors of norm  $O(\sqrt{\epsilon})$. The weights $w_m$ are defined in the same way as in the proof of Lemma \ref{lem: weights}. The main difference is that now correlation $p_m$ is defined to be any correlation such that \\$w_m \cdot p_m(a\mod 2,b\mod 2|x,y) \approx_{O(\epsilon)}\bra{\psi}\Pi_{A_x}^a\otimes \Pi_{B_y}^b \ket{\psi}$ for all $a,b \in \{2m,2m+1\}, x,y \in \{0,1\}$ (note that with an exact equality $p_m$ would not be a well-defined correlation, but an $O(\epsilon)$ correction is enough for existence of such a correlation $p_m$). We argue similarly for $w'_m$ and $p_m'$. The case of $d$ odd is also similar. 
\end{proof}





\begin{proof}[Proof of Theorem \ref{thm: robust self test}]
We look at the case of $d$ even first. Using Lemma \ref{lem: weights approx}, we deduce
\begin{align}
\sum_{m=0}^{\frac{d}{2}-1} [\textnormal{CHSH}_m]_p \approx_{O(d\epsilon)} \sum_{m=0}^{\frac{d}{2}-1} w_m \cdot [\textnormal{CHSH}]_{p_m} &\leq 2\sqrt{2} \label{eq: 23} \\
\sum_{m=0}^{\frac{d}{2}-1}[\textnormal{CHSH}'_m]_p \approx_{O(d\epsilon)} \sum_{m=0}^{\frac{d}{2}-1} w'_m \cdot [\textnormal{CHSH}]_{p_m} &\leq 2\sqrt{2} \label{eq: 24}
\end{align}
Hence, 
\begin{align}
\sum_{m=0}^{\frac{d}{2}-1} [\textnormal{CHSH}_m]_p &\leq 2\sqrt{2} + O(d\epsilon) \label{eq: 25} \\
\sum_{m=0}^{\frac{d}{2}-1}[\textnormal{CHSH}'_m]_p & \leq 2\sqrt{2} + O(d\epsilon) \label{eq: 26}
\end{align}
It is straightforward to see that \eqref{eq: 25} and \eqref{eq: 26} imply the existence of constants $C', C''>0$ such that
\begin{itemize}
    \item for all $m$,
    $w_m \leq C' \sqrt{d\epsilon}$ OR $[\textnormal{CHSH}]_{p_m} \geq 2\sqrt{2} - C''\sqrt{d\epsilon}$, AND
    \item for all $m$,
    $w'_m \leq C' \sqrt{\epsilon}$ OR $[\textnormal{CHSH}]_{p'_m} \geq 2\sqrt{2} - C''\sqrt{\epsilon}$.
\end{itemize}
From here, we deduce approximate equations like $\|\Pi_{A_0}^{2m+1} \ket{\psi}\|^2 \approx_{O(\sqrt{d\epsilon})} w_{m} \cdot \frac12$ and $\|\Pi_{A_0}^{2m+1} \ket{\psi}\|^2 \approx_{O(\sqrt{d\epsilon})} w'_m \cdot \frac12$, and similar other equations as in the proof of Theorem \ref{thm: exact self-testing}. These follow from a robust self-testing bound on CHSH. Now, such approximate equations imply, by applying triangle inequalities, that, for all $m$, 
\begin{equation}
\label{eq: ws}
    w_m \approx_{O(d^{3/2}\sqrt{\epsilon})} w'_m \approx_{O(d^{3/2}\sqrt{\epsilon})} \frac{2}{d}
\end{equation}
It is clear that there exists a constant $C>0$ such that for $\epsilon < \frac{C}{d^3}$, it must be that for all $m$ $w_m > C' \sqrt{d\epsilon}$ and $w'_m > C' \sqrt{d\epsilon}$. Hence, for $\epsilon < \frac{C}{d^3}$, it is the case that, for all $m$,
\begin{equation}
\label{eq: bound}
  [\textnormal{CHSH}]_{p_m}, [\textnormal{CHSH}]_{p'_m} \geq 2\sqrt{2} - C''\sqrt{d\epsilon}.
\end{equation}
Finally, recall the form of correlation $p$ from Lemma \ref{lem: weights approx}. This, combined with \eqref{eq: ws} and \eqref{eq: bound} implies that $p$ is $O(d^2\sqrt{\epsilon})$-close to the ideal correlation $\tilde{p}$ defined by the measurements of Lemma \ref{lem: ideal corr} (i.e. for all $a,b,x,y,$ $p(a,b|x,y) \approx_{O(d^2\sqrt{\epsilon})} \tilde{p}(a,b,|x,y)$, where $\tilde{p}$ is the ideal correlation). The robust self-testing statement from \cite{CS2017b} for the ideal correlation of Lemma \ref{lem: ideal corr} states that a strategy producing a correlation that is $\epsilon$-close to ideal, must be $O(d^3\epsilon^{\frac14})$-close (in the sense of Theorem \ref{thm: robust self test}) to the ideal strategy. Applying this to our analysis yields the conclusion of Theorem \ref{thm: robust self test}.

The case of $d$ odd is handled similarly.

\end{proof}

\end{appendix}

\end{document}